\documentclass[letterpaper, 11pt]{article}
\usepackage{fullpage}

\usepackage{epsfig,graphicx,graphics,color}
\usepackage{amsmath,amsthm,amssymb,color,comment,url}
\usepackage{mathtools,thmtools}
\usepackage{thm-restate}
\usepackage{caption}
\usepackage[position=b]{subcaption}
\usepackage[unicode]{hyperref}
\hypersetup{
    colorlinks=true,       
    linkcolor=blue,        
    citecolor=red,         
    filecolor=magenta,     
    urlcolor=cyan,         
    linktocpage=true
}
\usepackage[sort,nocompress]{cite}
\usepackage{titling} 
\thanksmarkseries{alph}
\usepackage{enumerate}

\def\final{0}  
\def\iflong{\iffalse}
\ifnum\final=0  
\usepackage[dvipsnames]{xcolor}
\newcommand{\yu}[1]{{\color{red}[{\tiny \textbf{Yu:} \bf #1}]\marginpar{\color{red}*}}}
\newcommand{\yutaro}[1]{{\color{red}[{\tiny \textbf{Yutaro:} \bf #1}]\marginpar{\color{red}*}}}
\newcommand{\tamas}[1]{{\color{red}[{\tiny \textbf{Tam\'as:} \bf #1}]\marginpar{\color{red}*}}}
\newcommand{\kristof}[1]{{\color{red}[{\tiny \textbf{Krist\'of:} \bf #1}]\marginpar{\color{red}*}}}
\else 
\newcommand{\yu}[1]{}
\newcommand{\yutaro}[1]{}
\newcommand{\tamas}[1]{}
\newcommand{\kristof}[1]{}
\fi

\numberwithin{equation}{section}

\theoremstyle{plain}
\newtheorem{theorem}{Theorem}[section]
\newtheorem{lemma}[theorem]{Lemma}

\newtheorem{claim}[theorem]{Claim}

\theoremstyle{definition}
\newtheorem{definition}[theorem]{Definition}

\newtheorem{question}[theorem]{Question}

\newcommand{\RR}{{\mathbb{R}}}
\newcommand{\ZZ}{{\mathbb{Z}}}
\newcommand{\bM}{{\mathbf{M}}}
\newcommand{\bN}{{\mathbf{N}}}
\newcommand{\cF}{{\mathcal{F}}}
\newcommand{\cH}{{\mathcal{H}}}
\newcommand{\cI}{{\mathcal{I}}}
\newcommand{\cO}{{\mathcal{O}}}

\newcommand{\cl}{\mathrm{cl}}
\newcommand{\rmax}{{r_\mathrm{max}}}
\newcommand{\rmin}{{r_\mathrm{min}}}
\newcommand{\rsum}{{r_\mathrm{sum}}}

\DeclareMathOperator\dist{dist}

\makeatletter
\def\namedlabel#1#2{\begingroup
    #2%
    \def\@currentlabel{#2}%
    \phantomsection\label{#1}\endgroup
}
\makeatother

\usepackage[algo2e,boxruled,linesnumbered,procnumbered]{algorithm2e}
\makeatletter
\renewcommand{\algocf@caption@boxruled}{%
  \hrule
  \hbox to \hsize{%
    \vrule\hskip-0.4pt
    \vbox{   
       \vskip\interspacetitleboxruled%
       \unhbox\algocf@capbox\hfill
       \vskip\interspacetitleboxruled
       }%
     \hskip-0.4pt\vrule%
   }\nointerlineskip%
}%
\makeatother

\usepackage{etoolbox}
\newcounter{algoline}

\AtBeginEnvironment{algorithm}{\setcounter{algoline}{0}}

\title{Matroid Intersection under Restricted Oracles}
\thanksmarkseries{alph}
\author{
Krist\'of B\'erczi\thanks{ELKH-ELTE Egerv\'ary Research Group and MTA-ELTE Momentum Matroid Optimization Research Group, Department of Operations Research, E\"otv\"os Lor\'and University, Budapest, Hungary. Email: \texttt{kristof.berczi@ttk.elte.hu, tamas.kiraly@ttk.elte.hu}}
\and
Tam\'as Kir\'aly\footnotemark[1]
\and
Yutaro Yamaguchi\thanks{Department of Information and Physical Sciences, Graduate School of Information Science and Technology, Osaka University, Osaka, Japan. Email: \texttt{yutaro.yamaguchi@ist.osaka-u.ac.jp}}
\and
Yu Yokoi\thanks{Principles of Informatics Research Division, National Institute of Informatics, Tokyo, Japan. Email: \texttt{yokoi@nii.ac.jp}}
}

\date{\empty}

\begin{document}
\maketitle
\thispagestyle{empty}

\begin{abstract}
Matroid intersection is one of the most powerful frameworks of matroid theory that generalizes various problems in combinatorial optimization. Edmonds' fundamental theorem provides a min-max characterization for the unweighted setting, while Frank's weight-splitting theorem provides one for the weighted case. Several efficient algorithms were developed for these problems, all relying on the usage of one of the conventional oracles for both matroids.

In the present paper, we consider the tractability of the matroid intersection problem under restricted oracles. In particular, we focus on the rank sum, common independence, and maximum rank oracles. We give a strongly polynomial-time algorithm for weighted matroid intersection under the rank sum oracle. In the common independence oracle model, we prove that the unweighted matroid intersection problem is tractable when one of the matroids is a partition matroid, and that even the weighted case is solvable when one of the matroids is an elementary split matroid. Finally, we show that the common independence and maximum rank oracles together are strong enough to realize the steps of our algorithm under the rank sum oracle. 

\bigskip

\noindent \textbf{Keywords:} Matroid intersection, Tractability, Rank sum oracle, Common independence oracle, Maximum rank oracle
\end{abstract}

\clearpage
\setcounter{page}{1}

\section{Introduction}
\label{sec:introduction}

A cornerstone of matroid theory is the efficient solvability of the matroid intersection problem introduced by Edmonds \cite{edmonds1970submodular}. 
Efficient algorithms for weighted matroid intersection were developed subsequently by Edmonds \cite{edmonds1979matroid}, by Lawler \cite{lawler1970optimal,lawler1976combinatorial}, and by Iri and Tomizawa \cite{iri1976algorithm}.
The min-max duality theorem of Edmonds \cite{edmonds1970submodular} for the unweighted matroid intersection problem was generalized by Frank \cite{frank1981weighted} to the weighted case.
These results do not only provide a well-established framework that includes various tractable combinatorial optimization problems such as bipartite matching and arborescence packing, but in certain cases they are unavoidable in solving natural optimization problems that seem to be unrelated to matroids.
A beautiful example is the problem of computing a cheapest rooted $k$-connected spanning subgraph of a digraph \cite{frank2009rooted}.
This is a pure graph optimization problem and yet the only known polynomial algorithm is based on the recognition that minimal rooted $k$-connected subgraphs of a digraph form the common bases of two matroids, and therefore a weighted matroid intersection algorithm can be applied.

In order to design matroid algorithms and to analyze their complexity, it should be clarified how matroids are given.
As the number of bases can be exponential in the size of the ground set, defining a matroid in an explicit form is inefficient.
Rather than giving a matroid as an explicit input, it is usually assumed that one of the standard oracles is available, and the complexity of the algorithm is measured by the number of oracle calls and other elementary steps.
Another way to define a matroid is to give an explicit linear representation, but this restricts the scope of the algorithm to linear matroids for which such an explicit representation is known.

For both the unweighted and weighted problems, a variety of efficient algorithms have been developed; see e.g. \cite{edmonds1970submodular, edmonds1979matroid, lawler1970optimal, aigner1971matching, lawler1975matroid, iri1976algorithm, cunningham1986improved, frank1981weighted, brezovec1986two, huang2019exact}.
A common feature of these algorithms, and also all previous studies on matroid intersection, is that they assume the availability of one of the standard oracles for both matroids; e.g., we can ask for the rank of a subset in each of them.
Our main contribution is showing that this assumption is not necessary for the tractability of matroid intersection, not even in the weighted setting.

One motivation for studying restricted oracles comes from polymatroid matching, a framework introduced by Lawler~\cite{lawler1976combinatorial} as a common generalization of matroid intersection and non-bipartite matching.
In \cite{lovasz2009matching}, Edmonds' theorem was deduced from polymatroid matching using a sophisticated argument. The main point is that when the matroid intersection problem is formulated as a polymatroid matching problem, only the \emph{rank sum function} of the two matroids is used rather than the two rank functions.
Although the polymatroid matching problem cannot be solved in polynomial time in general~\cite{lovasz1981matroid, jensen1982complexity}, the hardness was shown through special instances that seem to be far from matroid intersection.
This suggests that matroid intersection might still be tractable when only the sum of the rank functions is available.

We mention that another natural oracle to consider is the \emph{minimum rank function}, which answers the smaller value among the two ranks of a subset.
It follows from the polyhedral results of Edmonds \cite{edmonds1970submodular} that this oracle suffices to describe the convex hull of common independent sets.
In an unpublished manuscript, Bárász \cite{egresqp-06-03} gave a polynomial-time algorithm for unweighted matroid intersection under the minimum rank oracle.
We will present additional results about this oracle in a separate paper~\cite{inpreparation}.

\paragraph{Our results}
Our goal is to settle the tractability of the weighted matroid intersection problem under restricted oracles. 
In particular, we will focus on three different oracles: rank sum, common independence, and maximum rank oracles.

The study of the rank sum oracle is motivated by the above discussed connection to polymatroid matching results.
The difficulty of giving an efficient algorithm is that the usual augmenting path approach cannot be applied directly, since the exchangeability graphs are not determined by the rank sum oracle.
Still, we are able to give a strongly polynomial-time algorithm for the weighted matroid intersection problem by emulating the Bellman--Ford algorithm without explicitly knowing the underlying graph. 

\begin{theorem}
\label{thm:ranksum}
There exists a strongly polynomial-time algorithm for the weighted matroid intersection problem in the rank sum oracle model.
\end{theorem}

It is not difficult to see that a common independence oracle can be constructed with the help of a rank sum oracle.
Therefore, any algorithm that is based on the usage of a common independence oracle immediately translates into an algorithm that uses a rank sum oracle.
Nevertheless, the reverse implication does not hold, hence the common independence oracle is strictly weaker.
We show that unweighted matroid intersection remains tractable under the common independence oracle model when one of the matroids is a partition matroid.

\begin{theorem}
\label{thm:ci}
There exists a strongly polynomial-time algorithm for the unweighted matroid intersection problem in the common independence oracle model when one of the matroids is a partition matroid with all-one upper bound on the partition classes.
\end{theorem}

Although the complexity of the problem in general remains an intriguing open question even for the unweighted setting, this seemingly simple case already includes matchings in bipartite graphs and arborescences.

Recently, Joswig and Schr\"oter~\cite{joswig2017matroids} introduced the notion of split matroids, a class with distinguished structural properties that generalizes paving matroids.
B\'erczi, Kir\'aly, Schwarcz, Yamaguchi and Yokoi~\cite{berczi2023hypergraph} showed that every split matroid can be obtained as the direct sum of a so-called elementary split matroid and uniform matroids, and provided a hypergraph characterization of elementary split matroids.
By relying on this characterization, we show that even weighted matroid intersection is tractable in the common independence oracle model when one of the matroids is from this class. 

\begin{theorem}
\label{thm:split}
There exists a strongly polynomial-time algorithm for the weighted matroid intersection problem in the common independence oracle model when one of the matroids is an elementary split matroid.
\end{theorem}

We will see that the maximum rank oracle does not carry too much information on its own.
However, when combined with the common independence oracle, they are strong enough to mimic our algorithm for the rank sum case.

\begin{theorem}
\label{thm:cimax}
There exists a strongly polynomial-time algorithm for the weighted matroid intersection problem when both the common independence and the maximum rank oracles are available.
\end{theorem}

\paragraph{Organization}
The rest of the paper is organized as follows.
Basic definitions and notation are introduced in Section~\ref{sec:preliminaries}, together with some fundamental results on matroid intersection.
Section~\ref{sec:oracles} describes the relation between different oracles.
We present our strongly polynomial algorithm under the rank sum oracle in Section~\ref{sec:rank_sum}.
The common independence oracle case when one of the matroids is a partition matroid or an elementary split matroid, as well as the combination of the common independence and maximum rank oracles, is discussed in Section~\ref{sec:ci}. 
\section{Preliminaries}
\label{sec:preliminaries}
For the basics on matroids and the matroid intersection problem, we refer the reader to \cite{oxley2011matroid, schrijver2003combinatorial}. Throughout the paper, for $i=1,2$, let $\bM_i=(E,\cI_i)$ be loopless\footnote{The assumption that the matroids are loopless is not restrictive as loops can be easily detected by any of the oracles considered.} matroids on the same finite ground set $E$ of size $n$, whose \textbf{independent set families}, \textbf{rank functions}, and \textbf{closure operators} are denoted by $\cI_i$, by $r_i$, and by $\cl_i$, respectively; that is, $r_i(X) = \max \left\{\, |Y| \mid Y \subseteq X,~Y \in \cI_i \,\right\}$ and $\cl_i(X) = \{\, e \in E \mid r_i(X \cup \{e\}) = r_i(X) \,\}$ for each $X \subseteq E$.
For two sets $X,Y\subseteq E$, we denote their \textbf{symmetric difference} by $X\triangle Y=(X\setminus Y)\cup(Y\setminus X)$. The \textbf{$k$-truncation} of a matroid $\bM=(S,\cI)$ is a matroid $(\bM)_{k}=(S,\cI^{\le k})$ with $\cI^{\le k}=\{\, X\in\cI \mid |X|\leq k \,\}$.
For $I \in \cI_i$ and $x \in \cl_i(I) \setminus I$, the \textbf{fundamental circuit} of $x$ with respect to $I$ in $\bM_i$ is denoted by $C_i(I, x) = \{\, y \in I \mid I + x - y \in \cI_i \,\}$. 

We consider four oracles for matroid intersection. Given a set $X\subseteq E$ as an input, a \textbf{rank sum oracle} (\textsc{\textsc{Sum}}) answers the sum $\rsum(X)\coloneqq r_1(X)+r_2(X)$ of the ranks of $X$, a \textbf{minimum rank oracle} (\textsc{Min}) answers the minimum $\rmin(X)\coloneqq\min\left\{r_1(X),r_2(X)\right\}$ of the ranks of $X$, a \textbf{maximum rank oracle} (\textsc{Max}) answers the maximum $\rmax(X)\coloneqq\max\left\{r_1(X),r_2(X)\right\}$ of the ranks of $X$, and a \textbf{common independence oracle} (\textsc{CI}) answers \emph{``Yes''} if $X\in\cI_1\cap\cI_2$ and \emph{``No''} otherwise.

Let us first overview some basic results on unweighted matroid intersection. In \cite{edmonds1970submodular}, Edmonds gave the following characterization for the maximum cardinality of a common independent set of two matroids.
\begin{theorem}[Edmonds~\cite{edmonds1970submodular}]\label{thm:Edmonds}
Given two matroids $\bM_1=(E,\cI_1)$ and $\bM_2=(E,\cI_2)$ on a common ground set $E$, the maximum cardinality of a common independent set of $\bM_1$ and $\bM_2$ is equal to
\begin{align*}
\min\left\{\, r_1(Z) + r_2(E \setminus Z) \mid Z \subseteq E \,\right\}.
\end{align*}
\end{theorem}

The notion of exchangeability graphs plays a central role in any matroid intersection algorithm.

\begin{definition}[Exchangeability Graphs]\label{def:exchange}
Assume that $I\in\cI_1\cap \cI_2$. The \textbf{exchangeability graph} corresponding to $I$ is a bipartite digraph $D[I] = (E \setminus I, I; A[I])$ defined as follows. Set
\begin{align*}
  S_I &\coloneqq \{\, s \in E \setminus I \mid I + s \in \cI_1 \,\},\\
  T_I &\coloneqq \{\, t \in E \setminus I \mid I + t \in \cI_2 \,\},
\end{align*}
where elements in $S_I$ and in $T_I$ are called \textbf{sources} and \textbf{sinks}, respectively. We then define the set $A[I] \coloneqq A_1[I] \cup A_2[I]$ of exchangeability arcs, where
\begin{align*}
  A_1[I] \coloneqq&\ \{\, (y, x) \mid x \in E \setminus I,~y \in I,~I + x - y \in \cI_1 \,\}\\
  =&\ \{\, (y, s) \mid s \in S_I,~y \in I \,\} \cup \{\, (y, x) \mid x \in E \setminus (I \cup S_I),~y \in C_1(I, x) \,\},\\[1mm]
  A_2[I] \coloneqq&\ \{\, (x, y) \mid x \in E \setminus I,~y \in I,~I + x - y \in \cI_2 \,\}\\
  =&\ \{\, (t, y) \mid t \in T_I,~y \in I \,\} \cup \{\, (x, y) \mid x \in E \setminus (I \cup T_I),~y \in C_2(I, x) \,\}.
\end{align*}
Note that $S_I$ and $A_1[I]$ depend only on $\cI_1$, and $T_I$ and $A_2[I]$ depend only on $\cI_2$.
\end{definition}

Brualdi \cite{brualdi1969comments} observed that the set $A_i[I]$ satisfies the following property for $i=1,2$.

\begin{lemma}\label{lem:UPM-inv}
	Let $I\in \cI_i$ and let $Z\subseteq E$ satisfy $|I\triangle Z|=|I|$ and $I\triangle Z\in \cI_i$. 
	Then $A_i[I]$ contains a perfect matching on $Z$
	(i.e., a set of vertex-disjoint arcs whose tails and heads constitute $Z$).
\end{lemma}

Krogdahl~\cite{krogdahl1974combinatorial,krogdahl1976combinatorial,krogdahl1977dependence} proved a partial converse to the above lemma.
\begin{lemma}[Unique Perfect Matching Lemma]\label{lem:UPM}
	Let $I\in \cI_i$ and let $Z\subseteq E$ satisfy $|I\triangle Z|=|I|$.
	If $A_i[I]$ contains a unique perfect matching on $Z$, then $I\triangle Z\in \cI_i$.
\end{lemma}

Finally, let us recall that a standard algorithm for finding a maximum-cardinality common independent set is driven by the following subroutine, Algorithm~\ref{alg:1} (see \cite[$\S$~41.2]{schrijver2003combinatorial}).

For any digraph $D=(E,A)$, a \textbf{path} in $D$ is a sequence $P=e_1e_2\cdots e_\ell$ of distinct vertices such that $(e_i,e_{i+1})\in A$ for each $i=1,2,\dots,\ell-1$; we call $P$ an \textbf{$e_1$--$e_\ell$ path} or \textbf{an $X$--$Y$ path} for sets $X \ni e_1$ and $Y \ni e_\ell$ to emphasize the end vertices, and define $\ell$ as the \textbf{length}.
A \textbf{cycle} in $D$ is a sequence $e_1e_2\cdots e_\ell e_1$ such that $e_1e_2 \cdots e_\ell$ is a path and $(e_\ell, e_1)\in A$.
We often identify a path or a cycle with its vertex set $\{e_1, e_2,\dots,e_\ell\}$.

\begin{algorithm2e}[ht]
\caption{{{\sc Augment}$[E, \cI_1, \cI_2, I]$}} \label{alg:1}
\SetAlgoLined
\SetKwInOut{Input}{Input}\SetKwInOut{Output}{Output}
\Input{A finite set $E$, oracle access to $\cI_1$ and $\cI_2$, and a common independent set $I \in \cI_1 \cap \cI_2$.}
\Output{A common independent set $J \in \cI_1 \cap \cI_2$ with $|J| = |I| + 1$ if one exists, or a subset $Z \subseteq E$ with $r_1(Z) + r_2(E \setminus Z) = |I|$ otherwise.}
\BlankLine

Construct the exchangeability graph $D[I]$ with source set $S_I$ and sink set $T_I$.

If some $t \in T_I$ is reachable from some $s \in S_I$, then find a shortest $S_I$--$T_I$ path $P$ in $D[I]$, and return $J = I \triangle P$. Otherwise, return $Z = \{\, e \in E \mid e~\text{can reach some}~t \in T_I~\text{in}~D[I] \,\}$.
\end{algorithm2e}

Now we turn to the weighted setting. For a weight function $w \colon E \to \RR$ and a subset $X \subseteq E$, define $w(X) \coloneqq \sum_{e \in X} w(e)$.
For a family $\cF \subseteq 2^E$, a subset $X \subseteq E$ is \textbf{$w$-maximal in $\cF$} if $X \in \textrm{arg\,max}\left\{\, w(Y) \mid Y \in \cF \,\right\}$.
We define $\cI_i^k \coloneqq \{\, X \in \cI_i \mid |X| = k \,\}$ for $i = 1, 2$ and $k = 0, 1, \dots, n$.

One approach to solve the weighted matroid intersection problem is via augmentation along cheapest paths in the exchangeability graph as shown in Algorithm~\ref{alg:2} (see~\cite[$\S$~41.3]{schrijver2003combinatorial}), where the cost function $c \colon E \to \RR$ is defined on the vertex set as follows:
  \begin{align}\label{eq:cost}
    c(e) &\coloneqq \begin{cases}
      w(e) & \text{if $e \in I$},\\
      -w(e) & \text{if $e \in E \setminus I$}.
    \end{cases}
  \end{align}
For each path (or cycle) $P$ in $D[I]$, we define the \textbf{cost} of $P$ as $c(P) \coloneqq \sum_{e \in P}c(e)$.

\begin{algorithm2e}[t]
	\caption{{{\sc CheapestPathAugment}$[E, w, \cI_1, \cI_2, I]$}} \label{alg:2}
	\SetAlgoLined	
	\SetKwInOut{Input}{Input}\SetKwInOut{Output}{Output}
\Input{A finite set $E$, a weight function $w \colon E \to \RR$, oracle access to $\cI_1$ and $\cI_2$, and a $w$-maximal set $I \in \cI_1^k \cap \cI_2^k$ for some $k\in\{0, 1, \dots, n - 1\}$.}
\Output{A $w$-maximal set $J \in \cI_1^{k+1} \cap \cI_2^{k+1}$ if one exists, or a subset $Z \subseteq E$ with $r_1(Z) + r_2(E \setminus Z) = |I|$ otherwise.}
\BlankLine

Construct the exchangeability graph $D[I]$ with source set $S_I$ and sink set $T_I$. In addition, define the cost function $c \colon E \to \RR$ by \eqref{eq:cost}.

If some $t \in T_I$ is reachable from some $s \in S_I$, then find a shortest cheapest $S_I$--$T_I$ path $P$ in $D[I]$ (i.e., the cost $c(P)$ is minimum, and subject to this, the length of $P$ is minimum), and return $J = I \triangle P$.
Otherwise, return $Z = \{\, e \in E \mid e~\text{can reach some}~t \in T_I~\text{in}~D[I] \,\}$.
\end{algorithm2e}

The next lemma characterizes $w$-maximal common independent sets in $\cI^k_1\cap \cI^k_2$.

\begin{lemma}[cf.~{\cite[Theorem~41.5]{schrijver2003combinatorial}}]\label{lem:negative_cycle}
A common independent set $I \in \cI_1^k \cap \cI_2^k$ is $w$-maximal if and only if $D[I]$ contains no negative-cost cycle with respect to  the cost function $c$ defined as \eqref{eq:cost}.
\end{lemma}

\section{Polynomial Reducibility of Oracles}
\label{sec:oracles}

The aim of this section is to clarify the relation between oracles for matroid intersection, which implies the relation between the problem under the restricted oracles.
For a single matroid, although there are many different types of oracles that are often used, many of these conventional oracles have the same computational power.
More precisely, we say that an oracle $\cO_1$ is \textbf{polynomially reducible} to another oracle $\cO_2$ if $\cO_1$ can be simulated by using a polynomial number of oracle calls to $\cO_2$ measured in terms of the size of the ground set.
Two oracles are \textbf{polynomially equivalent} if they are mutually polynomially reducible to each other.
In this sense, the rank, independence, strong basis, circuit-finding, spanning, and port oracles are polynomially equivalent~\cite{robinson1980computational,hausmann1981algorithmic,coullard1996independence}. 

As defined in Section~\ref{sec:preliminaries}, we consider four types of oracles for matroid intersection: \textsc{Sum}, \textsc{Min}, \textsc{Max}, and \textsc{CI}.
As it turns out, \textsc{Max} is not very useful on its own, but it provides a powerful tool when combined with any of the other three oracles.
We denote by a `+' sign when we have access to two of the oracles, e.g., \textsc{Min+Max} means that for a set $X\subseteq E$ the oracle tells both $\rmin(X)$ and $\rmax(X)$. 

In what follows, we discuss the reducibility of the oracles one by one.
In order to keep the presentation clear, some of the ideas appear multiple times.
For an overview of the results, see Figure~\ref{fig:oracles}. Observe that, by  $\rsum(X)=\rmin(X)+\rmax(X)$, any combination of at least two of \textsc{Min}, \textsc{Max}, and \textsc{Sum} is clearly equivalent.
This immediately implies that each of {\sc Min}, {\sc Sum}, and {\sc Max} is reducible to each of {\sc Min+Max}, {\sc Min+Sum}, and {\sc Sum+Max}.

\begin{figure}[t!]
\centering
\includegraphics[width=.5\linewidth]{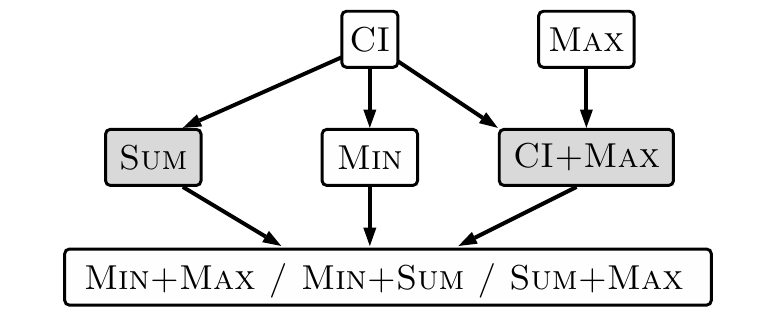}
\caption{Hierarchy of oracles, directed arcs representing polynomial reducibility. Grey boxes denote oracles for which strongly polynomial time algorithms are given in the present paper.}
\label{fig:oracles}
\end{figure}

\begin{theorem}\label{thm:or1} 
\textsc{CI} is not polynomially reducible to \textsc{Max}, but it is polynomially reducible to \textsc{Min} and \textsc{Sum}.
\end{theorem}
\begin{proof}
If $\bM_1$ is the free matroid, then \textsc{Max} always answers $|X|$ independently from the choice of $\bM_2$.
Thus deciding if $X\subseteq E$ is a common independent set or not is impossible relying solely on \textsc{Max}.

To see the second half, observe that for a set $X\subseteq E$, \textsc{CI} answers \emph{``Yes''} if and only if $X$ is a common independent set of the two matroids, that is, $r_1(X)=r_2(X)=|X|$.
By the subcardinality of the rank functions, this is equivalent to $\rmin(X)=|X|$ and to $\rsum(X)=2|X|$.
As these conditions can be checked by \textsc{Min} and \textsc{Sum}, respectively, the theorem follows.
\end{proof}

\begin{theorem}\label{thm:or2} 
\textsc{Min} is not polynomially reducible to \textsc{Sum}, \textsc{CI}, \textsc{Max}, and \textsc{CI+Max}.
\end{theorem}
\begin{proof}
We define two instances of the matroid intersection problem on the same ground set $E=\{a,b,c,d\}$ as follows.
For $i=1,2$, let $\bM_i$ be the graphic matroid of the graph $G_i$ on Figure~\ref{fig:m1m2a}, and let $\bM'_i$ be the graphic matroid of the graph $G'_i$ on Figure~\ref{fig:m1m2b}.
Consider the maximum-cardinality common independent set problem for $\bM_1$ and $\bM_2$, and for $\bM'_1$ and $\bM'_2$.
For any subset $X$ of $E$, both \textsc{Sum} and \textsc{CI} give the same answer in the two instances, thus it is not possible to distinguish them from each other using one of these oracles.
However, $\rmin(E)$ is $2$ in one of them while $3$ in the other one, showing that \textsc{Min} cannot be reduced to \textsc{Sum} or \textsc{CI}.

Now take the $3$-truncation of these graphic matroids, and define $\bN_1=(\bM_1)_3$, $\bN_2=(\bM_2)_3$, $\bN'_1=(\bM'_1)_3$, and $\bN'_2=(\bM'_2)_3$.
Consider the maximum-cardinality common independent set problem for $\bN_1$ and $\bN_2$, and for $\bN'_1$ and $\bN'_2$.
By the slight change in the definitions, both \textsc{CI} and \textsc{Max} give the same answer in the two instances for any subset $X\subseteq E$, thus it is not possible to distinguish them from each other using a combination of these two oracles.
However, $\rmin(E)$ is $2$ in one of them while $3$ in the other one, showing that \textsc{Min} cannot be reduced to \textsc{CI+Max}.

The case when $\bM_1$ is the free matroid shows again that $\rmin(X)$ cannot be determined relying solely on \textsc{Max}.
\end{proof}

\begin{figure}[t!]
\centering
\begin{subfigure}[t]{0.48\textwidth}
  \centering
  \includegraphics[width=.8\linewidth]{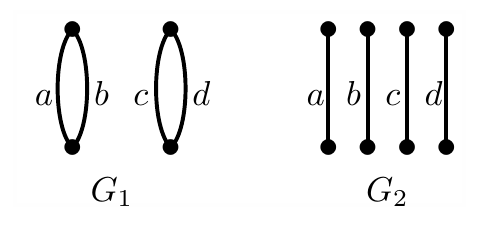}
  \caption{The graphs $G_1$ and $G_2$ defining $\bM_1$ and $\bM_2$.}
  \label{fig:m1m2a}
\end{subfigure}\hfill
\begin{subfigure}[t]{0.48\textwidth}
  \centering
  \includegraphics[width=.8\linewidth]{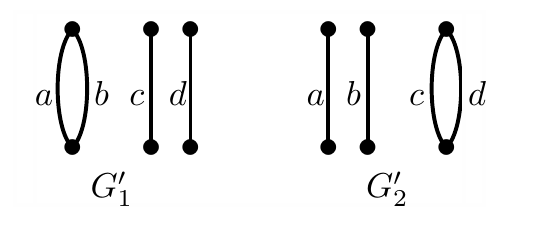}
  \caption{The graphs $G'_1$ and $G'_2$ defining $\bM'_1$ and $\bM'_2$.}
  \label{fig:m1m2b}
\end{subfigure}
\caption{Illustration of Theorems~\ref{thm:or2}, \ref{thm:or3}, and \ref{thm:or4}.}
\label{fig:or2}
\end{figure}

\begin{theorem}\label{thm:or3} 
\textsc{Sum} is not polynomially reducible to \textsc{Min}, \textsc{CI}, \textsc{Max}, and \textsc{CI+Max}.
\end{theorem}
\begin{proof}
If $\bM_1$ is a uniform matroid of rank $1$, then $\rmin(X)=0$ if $X=\emptyset$ and $1$ otherwise, while \textsc{CI} answers \emph{``Yes''} if $|X|\leq 1$ and \emph{``No''} otherwise.
These answers are independent from the choice of $\bM_2$, therefore we cannot determine $\rsum(X)$ with their help.

The case when $\bM_1$ is the free matroid shows again that $\rsum(X)$ cannot be determined relying solely on \textsc{Max}.

Consider the same two instances of the matroid intersection problem defined by the $3$-truncations of the graphic matroids on Figure~\ref{fig:or2} as in the proof of Theorem~\ref{thm:or2}.
Recall that both \textsc{CI} and \textsc{Max} give the same answer in the two instances for any subset $X\subseteq E$, thus it is not possible to distinguish them from each other using a combination of these two oracles.
However, $\rsum(E)$ is $5$ in one of them while $6$ in the other one, showing that \textsc{Sum} cannot be reduced to \textsc{CI+Max}.
\end{proof}

\begin{theorem}\label{thm:or4} 
\textsc{Max} is not polynomially reducible to \textsc{Min}, \textsc{Sum}, and \textsc{CI}.
\end{theorem}
\begin{proof}
The case when $\bM_1$ is a uniform matroid of rank $1$ shows again that $\rmax(X)$ cannot be determined relying solely on \textsc{Min} or \textsc{CI}.

Consider the same two instances of the matroid intersection problem defined by Figure~\ref{fig:or2} as in the proof of Theorem~\ref{thm:or2}.
Recall that for any subset $X$ of $E$, \textsc{Sum} gives the same answer in both instances, thus it is not possible to distinguish them from each other using \textsc{Sum}.
However, $\rmax(E)$ is $4$ in one of them while $3$ in the other one, showing that \textsc{Max} cannot be reduced to \textsc{Sum}.
\end{proof}
\newcommand{\nul}{{\mathsf{null}}}

\section{Matroid Intersection under Rank Sum Oracle}
\label{sec:rank_sum}

In Algorithm~\ref{alg:2}, we assume that we are given the independence oracles of matroids $\bM_1$ and $\bM_2$, which are polynomially equivalent to the oracles of rank functions $r_1$ and $r_2$, respectively. In this section, we consider the solvability of the weighted matroid intersection problem when only the \emph{rank sum function} $\rsum \colon 2^E \to \ZZ_{\geq 0}$ is available, that is, for any set $X\subseteq E$ the oracle answers the value of $\rsum(X) \coloneqq r_1(X)+r_2(X)$. Note that a subset $I \subseteq E$ belongs to $\cI_1\cap \cI_2$ if and only if $\rsum(I) = 2|I|$. However, when $\rsum(I)<2|I|$, we cannot decide whether $I\in \cI_i$ or not for each $i=1,2$. 

The matroid intersection problem under this oracle model is exactly a special case of the \emph{polymatroid matching problem}, which is equivalent to the so-called \emph{matroid matching (or parity) problem}~\cite[\S~11.1]{lovasz2009matching}. While a max-min duality theorem was given by Lov\'asz~\cite{lovasz1980matroid} for the case when the matroids in question admit linear representations\footnote{For the matroid intersection considered as a special case, the two matroids are necessarily representable over the same field.}, such a good characterization is not known in general even in the matroid intersection case.

In what follows, we provide an algorithm for the weighted matroid intersection problem with the rank sum oracle. We consider emulating Algorithm~\ref{alg:2}, i.e., {\sc CheapestPathAugment}$[E, w, \cI_1, \cI_2, I]$.

\subsection{Searching a Shortest Cheapest Path}
Take any $k \in \{0, 1, \dots, n - 1\}$ and let $I$ be a $w$-maximal set $\cI_1^k \cap \cI_2^k$.
To emulate Algorithm~\ref{alg:2}, we want to find a shortest cheapest $S_I$--$T_I$ path in 
$D[I] = (E \setminus I, I; A_1[I]\cup A_2[I])$ 
with respect to vertex cost $c\colon E \to \RR$ defined in Algorithm~\ref{alg:2}.
With only the rank sum oracle, however, we cannot determine the sets $A_1[I]$, $A_2[I]$, $S_I$, and $T_I$,
and hence cannot simply emulate Algorithm~\ref{alg:2}.

We show that, despite this difficulty, 
we can compute a shortest cheapest $S_I$--$T_I$ path in $D[I]$.
We first make some observations.
Let $D'[I]=(E \setminus I, I; A'_1[I]\cup A'_2[I])$ be the subgraph of $D[I]$ obtained from $D[I]$ by
removing arcs entering $S_I$ and leaving $T_I$, i.e.,  
\begin{align*}
A'_1[I] \coloneqq{}&\{\, (y, x) \mid x \in E \setminus (I \cup S_I),~y \in C_1(I, x) \,\},\\
A'_2[I] \coloneqq{}&\{\, (x, y) \mid x \in E \setminus (I \cup T_I),~y \in C_2(I, x) \,\}.
\end{align*}
Note that each element in $S_I\cap T_I$ is an isolated vertex in $D'[I]$.
Recall that $D[I]$ has no negative cost cycle with respect to $c$ by Lemma~\ref{lem:negative_cycle}.
This fact implies that finding a shortest cheapest $S_I$--$T_I$ path in $D[I]$ is equivalent to finding one in $D'[I]$. 

\begin{lemma}\label{lem:shortest-cheapest-path}
Any shortest cheapest $S_I$--$T_I$ path in $D'[I]$ is a shortest cheapest $S_I$--$T_I$ path in $D[I]$.
\end{lemma}

\begin{proof}
It is sufficient to show that any shortest cheapest $S_I$--$T_I$ path in $D[I]$ is contained in $D'[I]$.
Suppose, to the contrary, that a shortest cheapest $S_I$--$T_I$ path $P$ uses some arc $(y,s^*)\in A_1[I]\setminus A'_1[I]$ or $(t^*,y)\in A_2[I]\setminus A'_2[I]$, and let $s \in S_I$ and $t \in T_I$ be its end vertices.
If $P$ uses $(y,s^*)\in A_1[I]\setminus A'_1[I]$, let $P(s,y)$ and $P(s^*,t)$ be the subpaths of $P$ from $s$ to $y$ and from $s^*$ to $t$, respectively.
Since $(y,s)\in A_1[I]$, the path $P(s,y)$ is extended to a cycle with the same vertex set in $D[I]$, and hence $c(P(s,y))$ is nonnegative by Lemma~\ref{lem:negative_cycle}.
Then $c(P)=c(P(s,y))+c(P(s^*,t))\geq c(P(s^*,t))$ and $|P(s^*,t)|<|P|$, which contradicts that $P$ is a shortest cheapest $S_I$--$T_I$ path in $D[I]$.
If $P$ uses $(t^*,y)\in A_2[I]\setminus A'_2[I]$, we can similarly show that $P(s,t^*)$ is an $S_I$--$T_I$ path that is at least as cheap as $P$ and shorter than $P$.
\end{proof}

While we cannot determine $S_I$ and $T_I$, we can determine $S_I\cup T_I$ as $S_I\cup T_I=\{\,s\in E\setminus I\mid \rsum(I+s)\geq 2|I|+1\,\}$. We now provide a search algorithm with the rank sum oracle.
Its description is given as Algorithm~\ref{alg:3}.
For any $s\in S_I$ (resp., $s\in T_I$), it emulates the Bellman--Ford algorithm in $D'[I]$ 
(resp., in the inverse of $D'[I]$) rooted at $s$ without knowing $D'[I]$ explicitly.
Since there is no negative cost cycle in $D'[I]$ by Lemma~\ref{lem:negative_cycle}, the algorithm finds shortest cheapest paths from $s$ to all reachable vertices, and it returns a shortest cheapest $s$--$T_I$ path (resp., $s$--$S_I$ path) if it exists.
By applying this search algorithm for all $s\in S_I\cup T_I$, we can obtain 
a shortest cheapest $S_I$--$T_I$ path in $D'[I]$, 
which is also a shortest cheapest path in $D[I]$ by Lemma~\ref{lem:shortest-cheapest-path}.

For each $e\in E$, the algorithm maintains a sequence $P_e$ of distinct elements of $E$, which is initialized by a virtual token $\nul$ with $c(\nul) \coloneqq \infty$.
We will show that $P_e$ is an $s$--$e$ path in $D'[I]$ unless $P_e = \nul$.
For a sequence $P_e$ and an element $e'\in E\setminus P_e$, we denote by $P_e+e'$ the sequence obtained by appending $e'$ to $P_e$.

\begin{algorithm2e}[ht]
	\caption{{\sc EmulatingBellmanFord}$[E, c, \rsum, I, s]$} \label{alg:3}
	\SetAlgoLined	
	\SetKwInOut{Input}{Input}\SetKwInOut{Output}{Output}
\Input{A finite set $E$, a weight function $w \colon E \to \RR$, which defines $c \colon E \to \RR$ by \eqref{eq:cost}, oracle access to $\rsum \colon 2^E \to \ZZ_{\geq 0}$, 
	a $w$-maximal set $I \in \cI_1^k\cap \cI_2^k$ for some $k \in\{ 0, 1, \dots, n - 1\}$, and $s\in S_I\cup T_I$.}
\Output{For $s\in S_I$ (resp., $s\in T_I$), a shortest cheapest $s$--$T_I$ path (resp., $s$--$S_I$ path) in $D'[I]$ (resp, in the inverse of $D'[I]$) with respect to $c$ if one exists, or a message \emph{``No''} otherwise.}
\BlankLine

Set $P_s \gets s$, and $P_e \gets \nul$ for each $e\in E-s$. 

For $\ell=1,2,\dots,n-1$, do the following.
\protect{
\begin{minipage}{0.9\textwidth}
	\begin{description}
	\item[If $\ell$ is odd:] For each $y\in I$, do the following.
        \begin{itemize}
	\item Let $x\in E\setminus I$ minimize $c(P_x)$ subject to $P_x\neq \nul$, $y\not\in P_x$, and
        \begin{align*}
        (*)\quad\rsum(I\triangle P_x)=2|I|+1,\quad\rsum(I\triangle(P_x+y))=2|I|.
        \end{align*}
	\item If $c(P_x+y)<c(P_y)$, update $P_y\gets P_x+y$.
        \end{itemize}
	\item[If $\ell$ is even:] For each $x\in E\setminus I$, do the following.
	\begin{itemize}
        \item Let $y\in I$ minimize $c(P_y)$ subject to $P_y\neq \nul$, $x\not\in P_y$, and
	\begin{align*}
        (**)\quad&[~\rsum(I+x)=2|I|,~\rsum(I\triangle (P_y+x))=2|I|+1~]~\text{or}\\
		&[~\rsum(I+x)=2|I|+1,~ \rsum(I\triangle (P_y+x))=2|I|+2~].
        \end{align*}
	\item If $c(P_y+x)<c(P_x)$, update $P_x\gets P_y+x$.
        \end{itemize}
\end{description}
\end{minipage}
}

Let $t\in E\setminus I$ minimize $c(P_t)$ subject to $\rsum(I\triangle P_t)=2|I|+2$. 
Return $P_t$ if $c(P_t)\neq\infty$, and otherwise return \emph{``No''}. 
\end{algorithm2e}

\subsection{Correctness of the Search}
The following lemma shows an important property of Algorithm~\ref{alg:3}, where we can assume 
that $s$ belongs to $S_I = \{\, s \in E \setminus I \mid I + s \in \cI_1 \,\}$ by symmetry.
(For $s\in T_I = \{\, s \in E \setminus I \mid I + s \in \cI_2 \,\}$, replace $D'[I]$ with its inverse in the arguments.)

\begin{lemma}\label{lem:invariants2}
Let $s\in S_I$ and take any $\ell=1,2,\dots,n-1$. For any $e\in E$,
just after the $\ell$th updating process of Step~2 of Algorithm~\ref{alg:3},
the sequence $P_e$ is a shortest cheapest $s$--$e$ path in $D'[I]$ subject to $|P_e|\leq \ell+1$,
where $P_e=\nul$ means that there is no such path.
\end{lemma}

\begin{proof}
We use induction on $\ell$.
Note that the statement holds if $\ell=0$. 
We show the statement for any $\ell>0$ assuming that it holds for $\ell-1$. 
We use the following two claims.

\begin{claim}\label{claim:equivalence1}
Suppose that, for any $e$, $P_e$ is a shortest cheapest $s$--$e$ path in $D'[I]$ subject to $|P_e|\leq \ell$.
Then, for any $y\in I$ and $x\in E\setminus I$ such that [$P_x\neq \nul$, $y\not\in P_x$, and $c(P_x+y)<c(P_y)$], condition $(\ast)$ holds if and only if $(x,y)\in A'_2[I]$. 
\end{claim}

\begin{claim}\label{claim:equivalence2}
Suppose that, for any $e$, $P_e$ is a shortest cheapest $s$--$e$ path in $D'[I]$ subject to $|P_e|\leq \ell$.
Then, for any $x\in E\setminus I$ and $y\in I$ such that [$P_y\neq\nul$, $x\not\in P_y$, and $c(P_y+x)<c(P_x)$], condition $(\ast\ast)$ holds if and only if $(y,x)\in A'_1[I]$. 
\end{claim}

We postpone the proofs of these claims and complete the proof of the lemma relying on them.
For any $e\in E$, let $P^{\ell-1}_e$ and $P^{\ell}_e$ be the sequence $P_e$ just after the 
$(\ell-1)$st and $\ell$th process, respectively.
By induction, $P^{\ell-1}_e$ is a shortest cheapest $s$--$e$ path in $D'[I]$ subject to $|P_e|\leq \ell$.
Let $P^*_e$ be any shortest cheapest $s$--$e$ path in $D'[I]$ subject to $|P^*_e|\leq \ell+1$.

By the ``only if'' parts of Claims~\ref{claim:equivalence1} and~\ref{claim:equivalence2},
$P^{\ell}_e$ is an $s$--$e$ path in $D'[I]$ with $|P^{\ell}_e|\leq \ell+1$,
and hence $c(P^*_e)\leq c(P^{\ell}_e)$. 
Also, $c(P^{\ell}_e)\leq c(P^{\ell-1}_e)$ by the algorithm.
If $c(P^*_e)=c(P^{\ell-1}_e)$, then $P^{\ell}_e=P^{\ell-1}_e$ and the statement immediately follows.
Otherwise, $c(P^*_e)<c(P^{\ell-1}_e)$. This implies $|P^*_e|=\ell+1$.
Then $e\in I$ if $\ell$ is odd and $e\in E\setminus I$ if $e$ is even
(recall that $D'[I]$ is a bipartite digraph between $E\setminus I$ and $I$).
Let $e'$ be the second last element in $P^*_e$ and let $P^*_{e'}:=P^*_{e}-e$ (i.e., delete $e$ from $P^*_e$). 
Then $P^*_{e'}$ is an $s$--$e'$ path with $|P^*_{e'}|=\ell$, 
and hence $c(P^{\ell-1}_{e'})\leq c(P^*_{e'})$ by the induction hypothesis. 
So $c(P^{\ell-1}_{e'}+e)\leq c(P^*_e)<c(P^{\ell-1}_e)$.
If $\ell$ is odd, then $(e',e)\in A'_2[I]$, and hence ($\ast$) holds with $x:=e'$ and $y:=e$ by Claim~\ref{claim:equivalence1}.
If $\ell$ is even, then $(e',e)\in A'_1[I]$, and hence ($\ast\ast$) holds with $y:=e'$ and $x:=e$ by Claim~\ref{claim:equivalence2}.
In either case, we obtain $c(P^\ell_e)\leq c(P^*_{e'}+e)=c(P^*_e)$ and $|P^{\ell}_e|\leq \ell+1=|P^*_e|$.
Hence $P^\ell_e$ is a shortest cheapest $s$--$e$ path subject to $|P_e|\leq \ell+1$. 
\end{proof}

In what follows, we prove Claims~\ref{claim:equivalence1} and~\ref{claim:equivalence2}.
First, we need the following lemma.

\begin{lemma}\label{lem:invariants1}
	Let $s\in S_I$. At any moment of the algorithm, the following conditions hold.
	\begin{enumerate}
		\item[\rm (a)] Any $y\in I$ with $P_y\neq \nul$ satisfies $|I\triangle P_y|=|I|$,
		$r_1(I\triangle P_y)=|I|$, and $r_2(I\triangle P_y)=|I|$.
		\item[\rm (b)] Any $x\in E\setminus I$ with $P_x\neq \nul$ satisfies 
		$|I\triangle P_x|=|I|+1$, $r_1(I\triangle P_x)=|I|+1$, and $r_2(I\triangle P_x)\geq|I|$.
		Moreover, it satisfies $r_2(I\triangle P_x)=|I|+1$ if and only if $x\in T_I$.
		\item[\rm (c)] For each $e\in E$ with $P_e \neq \nul$, we have $(P_e-s)\cap S_I=\emptyset$ and
		$(P_e-e)\cap T_I=\emptyset$.
	\end{enumerate}
\end{lemma}

\begin{proof}
	By the algorithm, for each $e\in E$ with $P_e\neq \nul$, 
	the sequence $P_e$ starts with $s\in E\setminus I$ and 
	uses elements in $E\setminus I$ and $I$ alternately. 
	Then $|I\triangle P_y|=|I|$ for any $y\in I$ with $P_y\neq\nul$ and 
	$|I\triangle P_x|=|I|+1$ for any $x\in E\setminus I$ with $P_x\neq\nul$.
	For any $y\in I$, after $P_y$ is updated, it satisfies $\rsum(I\triangle P_y)=2|I|$ by 
	the condition ($\ast$) for update. Then (a) follows.
	
	For any $e\in E\setminus I$ with $P_e \neq \nul$, any $x'\in (P_e\setminus I)-e$ has some succeeding element $y'\in I$ in $P_e$ 
	and ($\ast$) holds for $x'$ and $y'$. Hence $\rsum(I\triangle P_{x'})=2|I|+1$.
	If $x'=s\in S_I$, it immediately implies $x'\not\in T_I$.
	If $x'\neq s$, then $x'$ has some preceding element $y''$ in $P_x$,
	and ($\ast\ast$) for $y''$ and $x'$ implies $\rsum(I+x')=2|I|$ (as $\rsum(I \triangle P_{x'}) \neq 2|I| + 2$), and hence $x'\not\in S_I\cup T_I$.
	Thus, $s\in S_I\setminus T_I$ and any $x'\in (P_e\setminus I)-s-e$ satisfies $x'\not\in S_I\cup T_I$. 
	
	For any $x\in (E\setminus I)-s$ with $P_x\neq \nul$, by ($\ast\ast$) for $x$ and its preceding element $y$, we have
	$[\rsum(I+x)=2|I|,\ \rsum(I\triangle P_x)=2|I|+1]$ or 
	$[\rsum(I+x)=2|I|+1,\ \rsum(I\triangle P_x)=2|I|+2]$.
	In the former case, $\rsum(I+x)=2|I|$ implies $x\not\in S_I\cup T_I$, 
	and hence $P_x\cap T_I=\emptyset$.
        This implies $r_2(I\triangle P_x)\leq r_2(I \cup P_x) = |I|$, and then $\rsum(I\triangle P_x)=2|I|+1$ implies $r_1(I\triangle P_x)=|I|+1$ and $r_2(I\triangle P_x)=|I|$.
	In the latter case, $\rsum(I\triangle P_x)=2|I|+2$ implies 
	$r_1(I\triangle P_x)=r_2(I\triangle P_x)=|I|+1$.
        Since any $x' \in (P_x \setminus I) - x$ satisfies $x' \not\in T_I$ (as seen in the previous paragraph), we must have $x\in T_I$,
	and then $x\not\in S_I$ follows from $\rsum(I+x)=2|I|+1$.
	Thus, (b) and (c) are shown.
\end{proof}

Now we are ready to show Claims~\ref{claim:equivalence1} and~\ref{claim:equivalence2}.
We sometimes denote by $V(P)$ the set of elements in a sequence $P$ for emphasizing that we focus on the set rather than the sequence.

\begin{proof}[Proof of Claim~\ref{claim:equivalence1}]
By Lemma~\ref{lem:invariants1}(b),
$\rsum(I\triangle P_x)=2|I|+1$ is equivalent to $x\not\in T_I$.
Lemma~\ref{lem:invariants1}(b) also implies $r_1(I\triangle P_x)=|I|+1=|I\triangle P_x|$, 
and hence $r_1(I\triangle (P_x+y))=r_1((I\triangle P_x)-y)=|I|$.
Therefore, $\rsum(I\triangle (P_x+y))=2|I|$ is equivalent to 
$r_2(I\triangle (P_x+y))=|I|=|I\triangle (P_x+y)|$, i.e., $I\triangle (P_x+y)\in \cI_2$.
Then, the condition $(\ast)$ is equivalent to
\begin{align*}
(*)'\quad x \not\in T_I,\quad I \triangle (P_x + y) \in \cI_2.
\end{align*}
We show that $(\ast)'$ holds if and only if $(x,y)\in A'_2[I]$.
By the induction hypothesis, $P_x$ is an $s$--$x$ path in $D'[I]$, 
and hence it uses arcs of $A'_2[I]$ and $A'_1[I]$ alternately.
Let $N_1$ and $N_2$ be the sets of those arcs of $A'_1[I]$ and $A'_2[I]$, respectively.
Since $x\in E\setminus I$, $N_1$ forms a matching that covers $V(P_x)-s$
and $N_2$ forms a matching that covers $V(P_x)-x$.

To show the ``if'' part, suppose $(x,y)\in A'_2[I]$. Then $x\not\in T_I$ by the definition of $A'_2[I]$.
Also, $N'_2:=N_2+(x,y)\subseteq A'_2[I]$ forms a perfect matching on $V(P_x+y)$.
Suppose conversely that $I\triangle (P_x+y)\not\in \cI_2$.
Then Lemma~\ref{lem:UPM} implies that $A_2[I]$ contains some other perfect matching $N''_2$ on $V(P_x+y)$.
We see that $N''_2\subseteq A'_2[I]$ because $(P_x+y)\cap T_I=\emptyset$ by Lemma~\ref{lem:invariants1}(c) and $x\not\in T_I$. Thus, $N'_2$, $N''_2$, and $N_1$ are all contained in $D'[I]$.
Consider the digraph $D=(V(P_x+y), A)$ whose arc set $A$ consists of the arcs in $N'_2$, $N''_2$, and two copies of $N_1$, where we consider their multiplicity, i.e., each arc in $(N'_2 \cap N''_2) \cup N_1$ is taken twice (parallel).
Since $N''_2\neq N'_2$, there exists an arc in $N''_2$ whose head precedes its tail on the path $P_x+y$.
Then $D$ contains at least one directed cycle.
The indegree and outdegree of each vertex in $D$ are given as
$(d^{\rm in}(s),d^{\rm out}(s))=(0,2)$, 
$(d^{\rm in}(y),d^{\rm out}(y))=(2,0)$,
and $(d^{\rm in}(e),d^{\rm out}(e))=(2,2)$ for all the other vertices $e$.
Then $A$ can be decomposed into the arc sets of two $s$--$y$ paths and one or more cycles.
Let $P_1$ and $P_2$ be those $s$--$y$ paths and $\mathcal{Q}$ be the set of those cycles
(where paths and cycles are sequences of vertices).
Then each vertex in $V(P_x+y)$ is used exactly twice in this decomposition, and hence 
\[\textstyle{c(P_1)+c(P_2)+\sum_{Q \in \mathcal{Q}}c(Q)=2c(P_x+y).}\]
Since $D'[I]$ has no negative cycle, this implies $c(P_1)+c(P_2)\leq 2c(P_x+y)$.
Also, $\mathcal{Q}\neq \emptyset$ implies $V(P_1)\subsetneq V(P_x+y)$ or $V(P_2)\subsetneq v(P_x+y)$.
In case $V(P_1)=V(P_x+y)$, we have $V(P_2)\subsetneq V(P_x+y)$, which implies $|P_2|<|P_x+y|\leq\ell+1$ 
because $|P_x|\leq \ell$ holds by induction.
Also, $V(P_1)=V(P_x+y)$ implies $c(P_2)\leq c(P_x+y)$, where $c(P_x+y)<c(P_y)$ by assumption.
Thus,  $P_2$ is an $s$--$y$ path in $D'[I]$ with $c(P_2)<c(P_y)$ and $|P_2|\leq \ell$, which 
contradicts the induction hypothesis that $P_y$ is a shortest cheapest $s$--$y$ path subject to $|P_y|\leq \ell$.
The case $V(P_2)=V(P_x+y)$ is similar.
In case $V(P_1)\subsetneq V(P_x+y)$ and $V(P_2)\subsetneq V(P_x+y)$, 
both $P_1$ and $P_2$ satisfy $|P_i|\leq \ell$ and at least one of them, say $P_i$, satisfies $c(P_i)\leq c(P_x+y)<c(P_y)$,
which again contradicts the induction hypothesis on $P_y$.

We next show the ``only if'' part. Let $(\ast)'$ hold.
By $I\triangle (P_x+y)\in \cI_2$, Lemma~\ref{lem:UPM-inv} implies that 
$A_2[I]$ contains a perfect matching $N'_2$ on $V(P_x+y)$.
Also, $N'_2\subseteq A'_2[I]$ because $(P_x+y)\cap T_I=\emptyset$ by Lemma~\ref{lem:invariants1}(c) and 
$x\not\in T_I$. Thus, $N_2$, $N'_2$, and $N_1$ are all contained in $D'[I]$.
Consider the digraph $D^*=(V(P_x+y), A^*)$ whose arc set $A^*$ consists of the arcs in $N_2$, $N'_2$, and two copies of $N_1$, where we consider their multiplicity as before.
Conversely, suppose that $(x,y)\not\in A'_2[I]$. Then $(x,y)\not\in N'_2$.
Since $N'_2$ covers $V(P_x+y)$, 
it has an arc whose tail is $x$ and whose head precedes $x$ in the path $P_x+y$.
Then $D^*$ contains at least one directed cycle.
Note that
$(d^{\rm in}(s),d^{\rm out}(s))=(0,2)$, 
$(d^{\rm in}(x),d^{\rm out}(x))=(2,1)$,
$(d^{\rm in}(y),d^{\rm out}(y))=(1,0)$,
and $(d^{\rm in}(e),d^{\rm out}(e))=(2,2)$ for all the other vertices $e$ in $D^*$.
Then $A^*$ can be decomposed into the arc sets of one $s$--$x$ path, 
one $s$--$y$ path, and one or more cycles.
Let $R_x$, $R_y$, and $\mathcal{Q}'$ be that $s$--$x$ path, $s$--$y$ path, and the set of cycles, respectively.
Then each vertex in $V(P_x+y)-y$ is used twice and $y$ is used once in this decomposition, and hence 
\[\textstyle{c(R_x)+c(R_y)+\sum_{Q \in \mathcal{Q}'}c(Q)=c(P_x)+c(P_x+y).}\]
Since $D'[I]$ has no negative cost cycle, $c(R_x)+c(R_y)\leq c(P_x)+c(P_x+y)$.
Also, $\mathcal{Q}\neq \emptyset$ implies $V(R_x)\subsetneq V(P_x)$ or $V(R_y)\subsetneq v(P_x+y)$.
If $V(R_x)=V(P_x)$, then $V(R_y)\subsetneq V(P_x+y)$, which implies $|R_y|<|P_x+y|\leq \ell+1$. 
Also, $V(R_x)=V(P_x)$ implies $c(R_y)\leq c(P_x+y)<c(P_y)$.
Thus, $R_y$ is an $s$--$y$ path with $c(R_y)<c(P_y)$ and $|R_y|\leq \ell$, 
which contradicts the induction hypothesis on $P_y$.
In case $V(R_y)=V(P_x+y)$, we have $V(R_x)\subsetneq V(P_x)$, which implies $|R_x|<|P_x|$.
Also, $V(R_y)=V(P_x+y)$ implies $c(R_x)\leq c(P_x)$.
Thus, $R_x$ is an $s$--$x$ path with $c(R_x)\leq c(P_x)$ and $|R_x|<|P_x|$,
which contradicts the induction hypothesis on $P_x$.
In case $V(R_x)\subsetneq V(P_x)$ and $V(R_y)\subsetneq V(P_x+y)$, we have $|R_x|<|P_x|$ and $|R_y|\leq \ell$. 
Also, we have $c(R_x)\leq c(P_x)$ or $c(R_y)\leq c(P_x+y)<c(P_y)$, and
each of them yields a contradiction.
\end{proof}

\begin{proof}[Proof of Claim~\ref{claim:equivalence2}]
Since $I\triangle (P_y+x)=(I\triangle P_y)+x$, Lemma~\ref{lem:invariants1}(a) implies 
$|I\triangle (P_y+x)|=|I|+1$, $r_1(I\triangle (P_y+x))\geq |I|$, and $r_2(I\triangle (P_y+x))\geq |I|$.
Also, by Lemma~\ref{lem:invariants1}(c), 
we have $P_y\cap T_I=\emptyset$ (which implies $\cl_2(I)=\cl_2(I\triangle P_y)$), 
and hence $r_2(I\triangle (P_y+x))=|I|+1$ holds if and only if $x\in T_I$.
If $x\not\in T_I$ (resp., $x\in T_I$), then 
$\rsum(I\triangle (P_y+x))=2|I|+1$ (resp., $\rsum(I\triangle (P_y+x))=2|I|+2$) 
is equivalent to $r_1(I\triangle (P_y+x))=|I|+1$,
and also $\rsum(I+x)=2|I|$ (resp., $\rsum(I+x)=2|I|+1$) 
is equivalent to $x\not\in S_I$. 
Therefore, $(\ast\ast)$ is equivalent to
\begin{align*}
(**)'\quad x\not\in S_I,\quad I\triangle (P_y+x)\in \cI_1.
\end{align*}
We show that $(\ast\ast)'$ holds if and only if $(y,x)\in A'_1[I]$.
By induction hypothesis, $P_y$ is an $s$--$y$ path in $D'[I]$, 
and hence it uses arcs of $A'_2[I]$ and $A'_1[I]$ alternately.
Let $N_1$ and $N_2$ be the sets of those arcs of $A'_1[I]$ and $A'_2[I]$, respectively.
Then $N_1$ forms a matching that covers $V(P_y)-y-s$
and $N_2$ forms a matching that covers $V(P_y)$.

To show the ``if'' part, suppose $(y,x)\in A'_1[I]$. Then $x\not\in S_I$ by the definition of $A'_1[I]$.
Also, $N'_1:=N_1+(y,x)\subseteq A'_1[I]$ forms a perfect matching on $V(P_y+x)-s$.
Suppose conversely that $I\triangle (P_y+x)\not\in \cI_1$.
It implies $I\triangle(P_y+x-s)\not\in\cI_1$ as follows.
Lemma~\ref{lem:invariants1}(c) implies $(P_y+x)\cap S_I=\{s\}$, and hence $\cl_1(I\triangle (P_y+x-s))\subseteq \cl_1(I)$.
If $\cl_1(I\triangle (P_y+x-s)) = \cl_1(I)$, then $\cl_1(I \triangle (P_y + x)) = \cl_1(I + s)$, and $I \triangle (P_y + x) \in \cI_1$ as $I + s \in \cI_1$, a contradiction.
Thus, we have $\cl_1(I\triangle (P_y+x-s)) \subsetneq \cl_1(I)$, which implies $I \triangle (P_y + x - s) \not\in \cI_1$.
Then Lemma~\ref{lem:UPM} implies that $A_1[I]$ contains some other perfect matching $N''_1$ on $V(P_y+x)-s$.
We see that $N''_1\subseteq A'_1[I]$ because $(P_y+x-s)\cap S_I=\emptyset$ by Lemma~\ref{lem:invariants1}(b) and $x\not\in S_I$. Thus, $N'_1$, $N''_1$, and $N_2$ are all contained in $D'[I]$.
Consider the digraph $D=(V(P_x+y), A)$ whose arc set $A$ consists of the arcs in $N'_1$, $N''_1$, and two copies of $N_2$, where we consider their multiplicity as before.
Similarly to the proof of Claim~\ref{claim:equivalence1}, we see that there exists an $s$--$y$ path $P$ in $D'[I]$ with $c(P)<c(P_x)$ and $|P|\leq \ell$, which contradicts the induction hypothesis on $P_x$.

We next show the ``only if'' part. Let $(\ast\ast)'$ hold.
By $I\triangle (P_y+x-s)\subseteq I\triangle (P_y+x) \in \cI_1$, Lemma~\ref{lem:UPM-inv} implies that 
$A_1[I]$ contains a perfect matching $N'_1$ on $V(P_y+x)-s$.
Also, $N'_1\subseteq A'_1[I]$ because $(P_y+x-s)\cap S_I=\emptyset$ by Lemma~\ref{lem:invariants1}(c) and 
$x\not\in S_I$. Thus, $N_1$, $N'_1$, and $N_2$ are all contained in $D'[I]$.
Consider the digraph $D^*=(V(P_x+y), A^*)$ whose arc set $A^*$ consists of the arcs in $N_1$, $N'_1$, and two copies of $N_2$, where we consider their multiplicity as before.
Then, similarly to the proof of Claim~\ref{claim:equivalence1}, we see that there exists an $s$--$y$ or $s$--$x$ path whose property contradicts the induction hypothesis on $P_y$ or $P_x$.
\end{proof}

\begin{lemma}\label{lem:BellmanFord}
	The output of {\sc EmulatingBellmanFord}$[E, c, \rsum, I, s]$ is always correct.
\end{lemma}
\begin{proof}
For any $e\in E$, a shortest cheapest $s$--$e$ path $P$ satisfies $|P|\leq n$.
Then, after the $(n-1)$st updating process, $P_e$ is indeed a shortest cheapest $s$--$e$ path
by Lemma~\ref{lem:invariants2}.
Also, when some path is returned, it is a shortest cheapest $s$--$T_I$ path
by Lemma~\ref{lem:invariants1}(b).
\end{proof}

\subsection{Matroid Intersection Algorithm under Rank Sum Oracle}
Using Algorithm~\ref{alg:3} as a subroutine, we can emulate {{\sc CheapestPathAugment}$[E, w, \cI_1, \cI_2, I]$} as Algorithm~\ref{alg:4}.

\begin{algorithm2e}[ht]
	\caption{{{\sc CheapestPathAugmentRankSum}$[E, w, \cI_1, \cI_2, I]$}} \label{alg:4}
	\SetAlgoLined	
	\SetKwInOut{Input}{Input}\SetKwInOut{Output}{Output}
	\Input{A finite set $E$, a weight function $w \colon E \to \RR$, 
		oracle access to $\rsum \colon 2^E \to \ZZ_{\geq 0}$, 
		and a $w$-maximal set $I \in \cI_1^k \cap \cI_2^k$ for some $k = 0, 1, \dots, n - 1$.}
	\Output{A $w$-maximal set $J \in \cI_1^{k+1} \cap \cI_2^{k+1}$ if one exists, or a message \emph{``No''}.}
	\BlankLine
	
	Determine the set $S_I\cup T_I=\{\,s\in E\setminus I\mid \rsum(I+s)\geq 2|I|+1\,\}$.
	Define a cost function $c \colon E \to \RR$ by \eqref{eq:cost}.
	
For each $s\in S_I\cup T_I$, apply {\sc EmulatingBellmanFord}$[E, c, \rsum, I, s]$.

If some path is returned in Step 2, then let $P$ be a shortest cheapest one among all returned paths, and return $J = I \triangle P$.
Otherwise, return a message \emph{``No''}.
\end{algorithm2e}

The following theorem concludes the section.

\begin{lemma}\label{lem:cpars}
The output of {\sc CheapestPathAugmentRankSum}$[E, w, \cI_1, \cI_2, I]$ is always correct.
\end{lemma}
\begin{proof}
The correctness of Algorithm~\ref{alg:4} immediately follows from Lemmas~\ref{lem:shortest-cheapest-path} and~\ref{lem:BellmanFord}.
\end{proof}

Lemma~\ref{lem:cpars} completes the proof of Theorem~\ref{thm:ranksum}.

\begin{proof}[Proof of Theorem~\ref{thm:ranksum}]
Starting from $I=\emptyset$, the size of the common independent set can be gradually increased using  {\sc CheapestPathAugmentRankSum}$[E, w, \cI_1, \cI_2, I]$ until $I$ becomes a common independent set of maximum cardinality.
The correctness of the algorithm follows from Lemma~\ref{lem:cpars}.
Note that, if we are asked to find a maximum-weight common independent set, then it suffices to output one with maximum weight among the obtained $w$-maximal common independent sets.
\end{proof}
\section{Matroid Intersection under Common Independence Oracle}
\label{sec:ci}
As discussed in Section~\ref{sec:oracles}, the common independence oracle is strictly weaker than the minimum rank and the rank sum oracles. As weighted matroid intersection turned out to be tractable for the rank sum oracle, the complexity of the problem under the common independence oracle is especially interesting. 

In what follows, we present an algorithm for the unweighted matroid intersection problem when one of the matroids is a partition matroid, and an algorithm for the weighted matroid intersection problem when one of the matroids is an elementary split matroid. We also show that the common independence oracle, when complemented with the rank oracle, is strong enough to design an algorithm similar to that for the rank sum case.

\subsection{Intersection with Partition Matroid}
\label{sec:cipart}
The aim of this section is to show that the unweighted matroid intersection problem is tractable under the common independence oracle when $\bM_1$ is a \textbf{partition matroid} with all-one upper bound on the partition classes, that is, when $\cI_1$ is represented as $\cI_1 = \{\, I \subseteq E \mid |I \cap E_i| \leq 1\ \text{for $i=1,\dots,q$}\,\}$ for some partition $E=E_1\cup\dots\cup E_q$. We will provide an algorithm that emulates Algorithm~\ref{alg:1}, i.e., {{\sc Augment}$[E, \cI_1, \cI_2, I]$}, using only the common independence oracle.

Take any $k = 0, 1, \dots, n - 1$ and let $I\in \cI_1^k \cap \cI_2^k$.
To emulate Algorithm~\ref{alg:1}, we want to find a shortest $S_I$--$T_I$ path in 
the exchangeability graph $D[I] = (E \setminus I, I; A_1[I]\cup A_2[I])$. 
With only the common independence oracle, however, we cannot construct $D[I]$, and cannot determine even $S_I$ or $T_I$.

Note that a shortest $S_I$--$T_I$ path in $D[I]$ never uses arcs entering sources or leaving sinks.
Therefore, finding a shortest $S_I$--$T_I$ path in $D[I]$ is equivalent to finding it in $D'[I]$,
where $D'[I]$ is the subgraph of $D[I]$ obtained by removing those arcs from $D[I]$ (it is used also in Section~\ref{sec:rank_sum}).
We now provide a search procedure, described as Algorithm~\ref{alg:part2}, 
that will be used as a subroutine for our augmentating procedure.
If a given element $s$ belongs to $S_I$, this search algorithm works like the breadth first search in $D'[I]$ rooted at $s$,
and returns a shortest $s$--$T_I$ path or certifies the nonexistence of such a path. 
 
In Algorithm~\ref{alg:part2}, for each $y\in I$, a sequence $P_y$ of distinct elements is defined.
In our analysis, $P_y$ will turn out to be a shortest $s$--$y$ path in $D'[I]$.
We use the notation $P_y+x$ to denote the sequence obtained by appending an element $x$ to $P_y$.

\begin{algorithm2e}[ht]
	\caption{{{\sc EmulatingBFS}$[E, \cI_1 \cap \cI_2, I, s]$}} \label{alg:part2}
	\SetAlgoLined	
	\SetKwInOut{Input}{Input}\SetKwInOut{Output}{Output}
\Input{Oracle access to $\cI_1 \cap \cI_2$ where $\cI_1$ is the independent set family of a partition matroid, a common independent set $I \in \cI^k_1 \cap \cI^k_2$, and an element $s \in E \setminus I$.}
\Output{A sequence $P$ with $I\triangle P\in \cI_1 \cap \cI_2$ and $|I\triangle P|=k+1$ if one exists, or a message \emph{``No''} otherwise.
In particular, if $s\in S_I$ and $D'[I]$ has an $s$--$T_I$ path, then a shortest $s$--$T_I$ path is returned.} 
\BlankLine

If $I+s\in \cI_1 \cap \cI_2$, halt with returning $s$.

For each $y\in I$, set $P_y\gets sy$ if $I+s-y\in \cI_1\cap \cI_2$, and $P_y \gets \nul$ otherwise.

For $\ell=1,2,\dots$, do the following.
\protect{
\begin{minipage}{0.9\textwidth}
	\begin{enumerate}
	\renewcommand{\labelenumi}{(\roman{enumi})}
		\item If there is no $y\in I$ with $|P_y|=2\ell$, halt with returning \emph{``No''}.
		\item If there exist $y'\in I$ and $x\in E\setminus I$ such that 
		$|P_{y'}|=2\ell$, $x\not\in P_{y'}$, $\{y',x\}\not\in \cI_1\cap\cI_2$, and $I\triangle(P_{y'}+x)\in \cI_1\cap \cI_2$, then 
		halt with returning $P_{y'}+x$.
		\item For each $y\in I$ with $P_y = \nul$, if there exist
		$y'\in I$ and $x\in E\setminus I$ such that 
		$|P_{y'}|=2\ell$, $x\not\in P_{y'}$, $\{y',x\}\not\in \cI_1\cap \cI_2$, and $I\triangle(P_{y'}+x+y)\in \cI_1\cap \cI_2$,
		then $P_y\gets P_{y'}+x+y$.
			\end{enumerate}
\end{minipage}
}
\end{algorithm2e}

By the algorithm, it is clear that the output is either a sequence $P$ with $I\triangle P\in \cI_1 \cap \cI_2$ and
$|I\triangle P|=k+1$ or a message \emph{``No''}.
Also, if $s\in S_I\cap T_I$, we see that $s$ itself is a shortest $s$--$T_I$ path and is returned at Step~1.
Therefore, we assume $s\in S_I\setminus T_I$ and show that a shortest $s$--$T_I$ path is returned if such a path exists.
We denote by $\dist(s,T_I)$ the length (i.e., the number of vertices) of a shortest $s$--$T_I$ path in $D'[I]$
and by $\dist(s,y)$ the length of a shortest $s$--$y$ path in $D'[I]$ for each $y\in I$.
Note that $\dist(s,T_I)$ is odd and $\dist(s,y)$ is even for any $y\in I$.

\begin{lemma}\label{lem:part}
Let $s\in S_I\setminus T_I$. The following hold for any $y\in I$ and any $\ell=1,2,\dots$.
\begin{enumerate}
\renewcommand{\labelenumi}{{\rm (\alph{enumi})}}
	\item $P_y$ is defined in Step~2 if and only if $\dist(s,y)=2$.
	If defined, it is a shortest $s$--$y$ path in $D'[I]$.
	\item A sequence is returned in the $\ell$th process of Step~3(ii) if and only if $\dist(s,T_I)=2\ell+1$.
	The returned sequence is a shortest $s$--$T_I$ path in $D'[I]$.
	\item $P_y$ is defined in the $\ell$th process of Step~3(iii)  if and only if $\dist(s,y)=2\ell+2<\dist(s,T_I)$. 
	If defined, it is a shortest $s$--$y$ path in $D'[I]$.
\end{enumerate}
\end{lemma}
\begin{proof}
For any $y\in I$, $\dist(s,y)=2$ means $(s,y)\in A'_2[I]$, which is equivalent to $I+s-y\in \cI_2$ as $s\not\in T_I$.
Since $s\in S_I$ implies $I+s-y\in \cI_1$, then $I+s-y\in \cI_1\cap \cI_2$ holds if and only if $(s,y)\in A'_2[I]$.
When $(s,y)\in A'_2[I]$, clearly $sy$ is a shortest $s$--$y$ path. Thus, (a) is shown.

We show (b) and (c) by induction on $\ell$. 
Suppose that they hold for $1,2,\dots,\ell-1$ and we are at the beginning of the $\ell$th process of Step~3. 
Then $\dist(s,T_I)\geq 2\ell+1$ because otherwise the algorithm has halted before.
Take any $y'\in I$ with $|P_{y'}|=2\ell$. Then
\begin{itemize}
    \item $P_{y'}$ is a shortest $s$--$y'$ path in $D'[I]$ by (a) and induction for (c);
    \item $\cl_2(I)=\cl_2(I\triangle P_{y'})$ because 
    \begin{itemize}
        \item $|I|=|I\triangle P_{y'}|$ and $I\triangle P_{y'}\in \cI_2$ hold by the algorithm (Steps 2 and 3(iii)), and
        \item $P_{y'}\cap T_I=\emptyset$ follows from $\dist(s,T_I)\geq 2\ell+1$.
    \end{itemize}
\end{itemize} 
As $P_{y'}$ is a path in $D'[I]$, it uses arcs in $A'_2[I]$ and $A'_1[I]$ alternately.
Let $N_1$ and $N_2$ be the sets of those arcs of $A'_1[I]$ and $A'_2[I]$, respectively.
By $s\in E\setminus I$ and $y'\in I$, then $N_1$ and $N_2$ form matchings that cover $V(P_{y'})-s-y'$ and $V(P_{y'})$, respectively (recall that we denote by $V(P)$ the set of elements in a sequence $P$ for emphasizing that we focus on the set rather than the sequence).

Take any $x\in E\setminus I$ with $x\not\in P_{y'}$.
The following claim completes the proof of (b).

\begin{claim}\label{claim:partition1}
$(y',x)\in A'_1[I]$ and $x\in T_I$ if and only if $\{y',x\}\not\in \cI_1\cap\cI_2$ and $I\triangle(P_{y'}+x)\in \cI_1\cap \cI_2$.
\end{claim}

\begin{proof}
For the ``only if'' part, suppose $(y',x)\in A'_1[I]$ and $x\in T_I$.
As $\bf{M}_1$ is a partition matroid with all-one upper bounds,
$(y',x)\in A'_1[I]$ means that $\{y',x\}$ is a circuit in $\bf{M}_1$, and hence $\{y',x\}\not\in \cI_1\cap\cI_2$.
Since $s\in S_I$ and $N_1+(y',x)$ forms a matching that covers $V(P_{y'}+x)-s$, we have $I\triangle(P_{y'}+x)\in \cI_1$.
(In $I\triangle(P_{y'}+x)$, each element in $I\cap (P_{y'}+x-s)$ is replaced by another element in the same partition class 
and $s$ comes from a partition class whose element is not used in $I$.)
Also, $\cl_2(I)=\cl_2(I\triangle P_{y'})$ and $x\in T_I$ imply $I\triangle (P_{y'}+x)=(I\triangle P_{y'})+x\in \cI_2$.
Thus, the ``only if'' part is shown.

For the ``if'' part, suppose $\{y',x\}\not\in \cI_1\cap\cI_2$ and $I\triangle (P_{y'}+x)\in \cI_1\cap\cI_2$.
As $\cl_2(I)=\cl_2(I\triangle P_{y'})$ holds, $I\triangle (P_{y'}+x)\in \cI_2$ implies $x\in T_I$.
Then $\{y',x\}\in \cI_2$, and hence $\{y',x\}\not\in \cI_1\cap\cI_2$ implies that $\{y',x\}$ is a circuit in $\bf{M}_1$.
Thus $(y',x)\in A'_1[I]$.
\end{proof}

Suppose that we are at the beginning of $\ell$th process of Step~3(iii).
Take $y'$ and $x$ as before and take any $y\in I$ such that $P_y$ is undefined.
Then $\dist(s,y)>2\ell$ by (a) and induction for (c).
Also $(y',x)\in A'_1[I]$ implies $x\not\in T_I$ 
since otherwise the algorithm has halted at Step~3 (ii).
The following claim completes the proof of (c).

\begin{claim}\label{claim:partition2}
Assume that $(y',x)\in A'_1[I]$ implies $x\not\in T_I$.
Then $(y',x)\in A'_1[I]$ and $(x,y)\in A'_2[I]$ if and only if 
$\{y',x\}\not\in \cI_1\cap\cI_2$ and $I\triangle(P_{y'}+x+y)\in \cI_1\cap \cI_2$.
\end{claim}

\begin{proof}
For the ``only if'' part, suppose $(y',x)\in A'_1[I]$ and $(x,y)\in A'_2[I]$.
Similarly to the proof of Claim~\ref{claim:partition1}, $(y',x)\in A'_1[I]$ implies 
$\{y',x\}\not\in \cI_1\cap\cI_2$ and $I\triangle(P_{y'}+x)\in \cI_1$,
and hence $I\triangle(P_{y'}+x+y)=(I\triangle(P_{y'}+x))-y\in \cI_1$.
Suppose, to the contrary, $I\triangle(P_{y'}+x+y)\not\in \cI_2$.
Since $N_2+(x,y)\subseteq A'_2[I]$ forms a perfect matching on $V(P_{y'}+x+y)$, 
Lemma~\ref{lem:UPM} implies that there exists some other perfect matching $N'_2\subseteq A_2[I]$ on $V(P_{y'}+x+y)$.
This $N'_2$ is included in $A'_2[I]$ because $(P_{y'}+x+y)\cap T_I=\emptyset$ follows from $P_{y'}\cap T_I=\emptyset$ and $x\not\in T_I$.
Then, $D'[I]$ contains an $s$--$y$ path with arcs in $N_1\cup N'_2$ and length at most $2\ell$, 
which contradicts $\dist(s,y)>2\ell$.

For the ``if'' part, suppose $\{y',x\}\not\in \cI_1\cap\cI_2$ and $I\triangle(P_{y'}+x+y)\in \cI_1\cap \cI_2$.
By Lemma~\ref{lem:UPM-inv}, $I\triangle(P_{y'}+x+y)\in\cI_2$ implies that  
$A_2[I]$ contains a perfect matching $N_2''$ on $V(P_{y'}+x+y)$. 
Conversely, suppose $(x,y)\not\in N''_2$. Then $(x^*,y)\in N_2''$ for some $x^*\in P_{y'}$,
and $P_{y'}\cap T_I=\emptyset$ implies $(x^*,y)\in A'_2[I]$. 
Hence $D'[I]$ has an $s$--$y$ path with length at most $2\ell$, a contradiction.
Thus, $(x,y)\in N''_2\subseteq A_2[I]$, which implies $x\in T_I$ or $(x,y)\in A'_2[I]$,
where the latter implies $y\in C_2(I,x)\not\subseteq \{y',x\}$.
Thus, in both cases, $\{y',x\}\in \cI_2$. Therefore, $\{y',x\}\not\in \cI_1\cap\cI_2$ implies $\{y',x\}\not\in \cI_1$,
and hence $y'\in C_1(I,x)$, and $(y',x)\in A'_1[I]$ follows.
By assumption, we then have $x\not\in T_I$, and hence $(x,y)\in A'_2[I]$.
\end{proof}

Thus, both (b) and (c) hold for $\ell$.
\end{proof}

Lemma~\ref{lem:part} completes the proof of correctness of Algorithm~\ref{alg:part2}.

\begin{lemma}\label{lem:EmulatingBFS}
	The output of {\sc EmulatingBFS}$[E, \cI_1 \cap \cI_2, I,s]$ is always correct.
\end{lemma}

\begin{proof}
If $\dist(s,T_I)=1$, i.e., if $s\in S_I\cap T_I$, then the algorithm correctly returns $s$ at Step 1.
If $\dist(s,T_I)=2\ell+1>1$, then $s\in S_I\setminus T_I$, and hence Lemma~\ref{lem:part}(b) implies that a shortest $s$--$T_I$ path is returned in the $\ell$th process of Step 3(ii).
\end{proof}

Using {\sc EmulatingBFS} (Algorithm~\ref{alg:part2}) as a subroutine, we design a procedure that emulates {\sc Augment}$[E, \cI_1, \cI_2, I]$
with the common independence oracle.
\begin{algorithm2e}[h!]
	\caption{{{\sc AugmentCommonIndependencePartition}$[E, \cI_1 \cap \cI_2, I]$}} \label{alg:part}
	\SetAlgoLined	
	\SetKwInOut{Input}{Input}\SetKwInOut{Output}{Output}
	\Input{Oracle access to $\cI_1 \cap \cI_2$ where $\cI_1$ is the independent set family of a partition matroid, and a common independent set $I \in \cI^k_1 \cap \cI^k_2$.}
	\Output{A common independent set $J \in \cI^{k+1}_1 \cap \cI^{k+1}_2$ if one exists, or a message \emph{``No''} otherwise.}
	\BlankLine
	
	If $I + x \in \cI_1 \cap \cI_2$ for some $x \in E \setminus I$, then halt with returning $J = I + x$. \label{st:part1}
	
	For each $s \in E \setminus I$, perform {\sc EmulatingBFS}$[E, \cI_1 \cap \cI_2, I, s]$. If some sequence $P$ is returned, 
	then return $J=I\triangle P$. Otherwise return with the message \emph{``No''}. \label{st:part2}
\end{algorithm2e}

\begin{lemma}\label{lem:common_indep_partition}
	The output of {\sc AugmentCommonIndependencePartition}$[E, \cI_1 \cap \cI_2, I]$ (Algorithm~\ref{alg:part}) is always correct.
\end{lemma}

\begin{proof}
The output in Step~\ref{st:part1} is clearly correct. As Step~\ref{st:part2} returns some sequence $P$ only if $I\triangle P$ is indeed a common independent set of size $k+1$, it suffices to show that if there exists a common independent set $J$ of size $k+1$,
then {\sc EmulatingBFS}$[E, \cI_1\cap\cI_2,I, s]$ returns a sequence for some $s \in E \setminus I$. 
By the correctness of {\sc Augment}$[E, \cI_1, \cI_2, I]$, the existence of such $J$ implies
that $D[I]$ has some $S_I$--$T_I$ path, and so does $D'[I]$. Then $D'[I]$ contains an $s$--$T_I$ path for some $s\in S_I$. 
For such $s$, {\sc EmulatingBFS}$[E, \cI_1\cap\cI_2,I, s]$ returns a shortest $s$--$T_I$ path in $D'[I]$.
\end{proof}

Lemma~\ref{lem:common_indep_partition} completes the proof of Theorem~\ref{thm:ci}.

\begin{proof}[Proof of Theorem~\ref{thm:ci}]
	Starting from $I\coloneqq\emptyset$, the size of the common independent set can be gradually increased using {\sc AugmentCommonIndependencePartition}$[E, \cI_1 \cap \cI_2, I]$ until $I$ becomes a common independent set of maximum cardinality. The correctness of the algorithm follows by Lemma~\ref{lem:common_indep_partition}.
\end{proof}

\subsection{Intersection with Elementary Split Matroid}
\label{sec:split}
Motivated by the study of matroid polytopes from a tropical geometry point of view, Joswig and Schr\"oter~\cite{joswig2017matroids} introduced the notion of \textbf{split matroids}. This class does not only generalize paving matroids, but it is closed both under duality and taking minors. B\'erczi, Kir\'aly, Schwarcz, Yamaguchi and Yokoi~\cite{berczi2023hypergraph} later observed that every split matroid can be obtained as the direct sum of a so-called \textbf{elementary split matroid} and uniform matroids. Elementary split matroids capture all the nice properties of connected split matroids, and is closed not only under duality and taking minors but also truncation. Motivated by representations of paving matroids by hypergraphs, they provided a hypergraph characterization of elementary split matroids as follows.

Let $E$ be a ground set of size at least $r$, $\cH=\{H_1,\dots, H_q\}$ be a (possibly empty) collection of subsets of $E$ (called \textbf{hyperedges}), and $r, r_1, \dots, r_q$ be nonnegative integers satisfying
\begin{align}
|H_i \cap H_j| &\le r_i + r_j -r &&\text{for $1 \le i < j \le q$,}\tag*{(H1)}\label{eq:h1}\\
|E\setminus H_i| + r_i &\ge r &&\text{for $i=1,\dots, q$.} \tag*{(H2)}\label{eq:h2}
\end{align}
Then $\cI=\{\, X\subseteq E\mid |X|\leq r,\ |X\cap H_i|\leq r_i\ \text{for $1\leq i \leq q$} \,\}$ forms the family of independent sets of a rank-$r$ matroid $M$ with rank function $r_M(Z)=\min\big\{r,|Z|,\min_{1\leq i\leq q}\{|Z\setminus H_i|+r_i\}\big\}$.
Matroids that can be obtained in this form are called \textbf{elementary split matroids}. 

We call a set $F\subseteq E$ \textbf{$H_i$-tight} or \textbf{tight with respect to $H_i$} if $|F\cap H_i|=r_i$. The following lemma shows that an independent set of size less than $r$ cannot be tight with respect to two different hyperedges.

\begin{lemma}\label{lem:tight}
Let $M$ be an elementary split matroid with representation $\cH=\{H_1,\dots,H_q\}$ and $r,r_1,\dots,r_q$, and let $F$ be a set of size less than $r$. Then $F$ is tight with respect to at most one of the hyperedges.
\end{lemma}

\begin{proof}
Suppose to the contrary that $F$ is both $H_i$- and $H_j$-tight. Then we get
\begin{align*}
    |H_i\cap H_j|
    {}&{}\geq  
    |F\cap H_i\cap H_j|
    = 
    |F\cap H_i|+|F\cap H_j|-|F\cap (H_i\cup H_j)|\\
    {}&{}\geq  
    r_i+r_j-|F|
    >
    r_i+r_j-r,
\end{align*}
contradicting~\ref{eq:h1}.
\end{proof}

Now we show that the weighted matroid intersection problem is tractable under the common independence oracle when $\bM_1$ is an \textbf{elementary split matroid}, that is, when $\cI_1$ can be represented as $\cI_1=\{\, X\subseteq E\mid |X|\leq r,\ |X\cap H_i|\leq r_i\ \text{for $1\leq i \leq q$} \,\}$ for some (possibly empty) hypergraph $\cH=\{H_1,\dots, H_q\}$ and nonnegative integers $r, r_1, \dots, r_q$ satisfying \ref{eq:h1} and \ref{eq:h2}. The proof is based on observing that the exchangeability graph has a special structure.

\begin{proof}[Proof of Theorem~\ref{thm:split}]
Suppose that we have oracle access to the common independent set family $\cI_1\cap \cI_2$ of two matroids on $E$, where $\cI_1$ belongs to an elementary split matroid. Consider a $w$-maximal set $I \in \cI_1^k \cap \cI_2^k$ for some $k\in\{0, 1, \dots, n - 1\}$. According to Algorithm~\ref{alg:2} and Lemma~\ref{lem:shortest-cheapest-path}, a $w$-maximal set $J\in\cI_1^{k+1}\cap\cI_2^{k+1}$, if exists, can be obtained in the form $J=I\triangle P$ where $P$ is a shortest cheapest $S_I$--$T_I$ path in $D'[I]$. 

\begin{claim}\label{cl:short}
If $\cI_1^{k+1}\cap\cI_2^{k+1}\neq\emptyset$, then there exists a shortest cheapest $S_I$--$T_I$ path in $D'[I]$ of length at most $3$.
\end{claim}
\begin{proof}
As $\cI_1^{k+1}\cap\cI_2^{k+1}\neq\emptyset$, there necessarily exists an $S_I$--$T_I$ path in $D'[I]$; let $P=e_1e_2\cdots e_\ell$ be a shortest cheapest one. As $I$ is not a basis of $\bM_1$, observe that $I+x\notin\cI_1$ for some $x\in E\setminus I$ if and only if there exists a hyperedge $H_i$ such that $I$ is $H_i$-tight and $x\in H_i$. By Lemma~\ref{lem:tight}, $I$ is tight with respect to at most one of the hyperedges, hence we get that $E\setminus (S_I\cup I)\subseteq H_i$. This also implies that $A'_1[I]=\{\,(y,x)\mid x\in H_i\setminus I,~y\in H_i\cap I\,\}$.

Assume that the length of the path is more than $3$. Then, by the above observation, both $(e_2, e_\ell)$ and $(e_{\ell-1}, e_3)$ exist in $D'[I]$. Therefore $C\coloneqq e_3e_4\cdots e_{\ell-1}e_3$ is a cycle in $D'[I]$, and $P'\coloneqq e_1e_2e_\ell$ is an $S_i$--$T_i$ path in $D'[I]$. By Lemma~\ref{lem:negative_cycle}, the cost of $C$ is nonnegative, and hence the cost of $P'$ is at most the cost of $P$, contradicting the choice of $P$.
\end{proof}

By Claim~\ref{cl:short}, a $w$-maximal member of $\cI_1^{k+1}\cap\cI_2^{k+1}$, if exists, can be found by checking every set $I'$ of size $k+1$ with $|I\triangle I'|\leq 2$. This concludes the proof of the theorem.
\end{proof}

\subsection{Complemented with Maximum Rank Oracle}
\label{sec:cimax}
When access is given to both a common independence and a maximum rank oracle, every step of Algorithms~\ref{alg:3} and \ref{alg:4}, i.e., {\sc EmulatingBellmanFord} and {\sc CheapestPathAugmentRankSum}, can be emulated and hence the weighted matroid intersection problem is solved as with Section~\ref{sec:rank_sum}. 

\begin{proof}[Proof of Theorem~\ref{thm:cimax}]
Suppose that we have oracle access to the common independent set family $\cI_1 \cap \cI_2$ and the maximum rank function $\rmax$ of two matroids on $E$ instead of that to $\rsum$. In {\sc EmulatingBellmanFord} and {\sc CheapestPathAugmentRankSum}, we ask the rank sum oracle the following types of questions:
\begin{enumerate}
\renewcommand{\labelenumi}{(\alph{enumi})}
\item whether $\rsum(I' + x) = 2|I'|$ or not for $I' \in \cI_1 \cap \cI_2$ and $x \in E \setminus I'$, \label{it:reda}
\item whether $\rsum(I' + x) = 2|I'| + 1$ or not for $I' \in \cI_1 \cap \cI_2$ and $x \in E \setminus I'$, \label{it:redb}
\item whether $\rsum(I' + x) = 2|I'|+2$ or not for $I' \in \cI_1 \cap \cI_2$ and $x \in E \setminus I'$, \label{it:redc} and
\item whether $\rsum(I') = 2|I'|$ or not for $I' \subseteq E$. \label{it:redd}
\end{enumerate}

These questions can be tested using the common independence and the maximum rank oracles together as follows. The answer to \eqref{it:reda} is \emph{``Yes''} if and only if $I' + x \notin \cI_1 \cap \cI_2$ and $\rmax(I' + x) = |I'|$, the answer to \eqref{it:redb} is \emph{``Yes''} if and only if $I' + x \not\in \cI_1 \cap \cI_2$ and $\rmax(I' + x) = |I'| + 1$, the answer to \eqref{it:redc} is \emph{``Yes''} if and only if $I' + x \in \cI_1 \cap \cI_2$, and the answer to \eqref{it:redd} is \emph{``Yes''} if and only if $I' \in \cI_1 \cap \cI_2$. 
\end{proof}

\section{Concluding Remarks}
In this paper, we have shown the tractability of unweighted/weighted matroid intersection problems under several restricted oracles.
The rank sum oracle or the combination of the common independence oracle and the maximum rank oracle is enough to solve the weighted matroid intersection problem in polynomial time (Theorems~\ref{thm:ranksum} and \ref{thm:cimax}).
Also, if one matroid is restricted to a partition matroid with all-one upper bound or to an elementary split matroid, the unweighted or weighted problem, respectively, can be solved only using the common independence oracle (Theorems~\ref{thm:ci} and \ref{thm:split}).

The following two big questions still remain, and our subsequent paper~\cite{inpreparation} is tackling the first question.

\begin{question}
Is there a strongly polynomial-time algorithm for the weighted matroid intersection problem in the minimum rank oracle model?
Or can we show the hardness?
\end{question}

\begin{question}
Is there a strongly polynomial-time algorithm for the unweighted/weighted matroid intersection problem in the common independence oracle model?
Or can we show the hardness?
\end{question}

\section*{Acknowledgments}
We are grateful to Yuni Iwamasa and Taihei Oki for initial discussions on the problem.
We would like to thank the reviewers, who read this paper and provided positive and helpful comments.

The work was supported by the Lend\"ulet Programme of the Hungarian Academy of Sciences -- grant number LP2021-1/2021 and by the Hungarian National Research, Development and Innovation Office -- NKFIH, grant numbers FK128673 and TKP2020-NKA-06.
Yutaro Yamaguchi was supported by JSPS KAKENHI Grant Numbers 20K19743 and 20H00605, and by Overseas Research Program in Graduate School of Information Science and Technology, Osaka University.
Yu Yokoi was supported by JST PRESTO Grant Number JPMJPR212B.

\bibliographystyle{abbrv}
\bibliography{MI_oracle}

\begin{thebibliography}{10}

\bibitem{aigner1971matching}
M.~Aigner and T.~A. Dowling.
\newblock Matching theory for combinatorial geometries.
\newblock {\em Transactions of the American Mathematical Society},
  158(1):231--245, 1971.

\bibitem{egresqp-06-03}
M.~B{\'a}r{\'a}sz.
\newblock Matroid intersection for the min-rank oracle.
\newblock Technical Report QP-2006-03, Egerv{\'a}ry Research Group, Budapest,
  2006.
\newblock {\tt http://www.cs.elte.hu/egres/}.

\bibitem{inpreparation}
M.~B{\'a}r{\'a}sz, K.~B{\'e}rczi, T.~Kir{\'a}ly, Y.~Yamaguchi, and Y.~Yokoi.
\newblock Matroid intersection under minimum rank oracle.
\newblock In preparation.

\bibitem{berczi2023hypergraph}
K.~B{\'e}rczi, T.~Kir{\'a}ly, T.~Schwarcz, Y.~Yamaguchi, and Y.~Yokoi.
\newblock Hypergraph characterization of split matroids.
\newblock {\em Journal of Combinatorial Theory, Series A}, 194:105697, 2023.

\bibitem{brezovec1986two}
C.~Brezovec, G.~Cornu{\'e}jols, and F.~Glover.
\newblock Two algorithms for weighted matroid intersection.
\newblock {\em Mathematical Programming}, 36(1):39--53, 1986.

\bibitem{brualdi1969comments}
R.~A. Brualdi.
\newblock Comments on bases in dependence structures.
\newblock {\em Bulletin of the Australian Mathematical Society}, 1(2):161--167,
  1969.

\bibitem{coullard1996independence}
C.~R. Coullard and L.~Hellerstein.
\newblock Independence and port oracles for matroids, with an application to
  computational learning theory.
\newblock {\em Combinatorica}, 16(2):189--208, 1996.

\bibitem{cunningham1986improved}
W.~H. Cunningham.
\newblock Improved bounds for matroid partition and intersection algorithms.
\newblock {\em SIAM Journal on Computing}, 15(4):948--957, 1986.

\bibitem{edmonds1970submodular}
J.~Edmonds.
\newblock Submodular functions, matroids, and certain polyhedra.
\newblock In {\em Combinatorial Structures and Their Applications}, pages
  69--87. Gorden and Breach, 1970.
\newblock (Also in \emph{Combinatorial Optimization --- Eureka, You Shrink!},
  pages~11--26, Springer, 2003.).

\bibitem{edmonds1979matroid}
J.~Edmonds.
\newblock Matroid intersection.
\newblock {\em Annals of Discrete Mathematics}, 4:39--49, 1979.

\bibitem{frank1981weighted}
A.~Frank.
\newblock A weighted matroid intersection algorithm.
\newblock {\em Journal of Algorithms}, 2(4):328--336, 1981.

\bibitem{frank2009rooted}
A.~Frank.
\newblock Rooted $k$-connections in digraphs.
\newblock {\em Discrete Applied Mathematics}, 157(6):1242--1254, 2009.

\bibitem{hausmann1981algorithmic}
D.~Hausmann and B.~Korte.
\newblock Algorithmic versus axiomatic definitions of matroids.
\newblock In {\em Mathematical Programming at Oberwolfach}, pages 98--111.
  Springer, 1981.

\bibitem{huang2019exact}
C.-C. Huang, N.~Kakimura, and N.~Kamiyama.
\newblock Exact and approximation algorithms for weighted matroid intersection.
\newblock {\em Mathematical Programming}, 177(1-2):85--112, 2019.

\bibitem{iri1976algorithm}
M.~Iri and N.~Tomizawa.
\newblock An algorithm for finding an optimal ``independent assignment".
\newblock {\em Journal of the Operations Research Society of Japan},
  19(1):32--57, 1976.

\bibitem{jensen1982complexity}
P.~M. Jensen and B.~Korte.
\newblock Complexity of matroid property algorithms.
\newblock {\em SIAM Journal on Computing}, 11(1):184--190, 1982.

\bibitem{joswig2017matroids}
M.~Joswig and B.~Schr{\"o}ter.
\newblock Matroids from hypersimplex splits.
\newblock {\em Journal of Combinatorial Theory, Series A}, 151:254--284, 2017.

\bibitem{krogdahl1974combinatorial}
S.~Krogdahl.
\newblock A combinatorial base for some optimal matroid intersection
  algorithms.
\newblock Technical Report STAN-CS-74-468, Computer Science Department,
  Stanford University, Stanford, CA, U.S., 1974.

\bibitem{krogdahl1976combinatorial}
S.~Krogdahl.
\newblock A combinatorial proof for a weighted matroid intersection algorithm.
\newblock Technical Report Computer Science Report 17, Institute of
  Mathematical and Physical Sciences, University of Tromso, Tromso, Norway,
  1976.

\bibitem{krogdahl1977dependence}
S.~Krogdahl.
\newblock The dependence graph for bases in matroids.
\newblock {\em Discrete Mathematics}, 19(1):47--59, 1977.

\bibitem{lawler1970optimal}
E.~L. Lawler.
\newblock Optimal matroid intersections.
\newblock In {\em Combinatorial Structures and Their Applications}, pages
  233--234. Gorden and Breach, 1970.

\bibitem{lawler1975matroid}
E.~L. Lawler.
\newblock Matroid intersection algorithms.
\newblock {\em Mathematical Programming}, 9(1):31--56, 1975.

\bibitem{lawler1976combinatorial}
E.~L. Lawler.
\newblock {\em Combinatorial Optimization: Networks and Matroids}.
\newblock Holt, Rinehart and Winston, 1976.

\bibitem{lovasz1980matroid}
L.~Lov{\'a}sz.
\newblock Matroid matching and some applications.
\newblock {\em Journal of Combinatorial Theory, Series B}, 28(2):208--236,
  1980.

\bibitem{lovasz1981matroid}
L.~Lov{\'a}sz.
\newblock The matroid matching problem.
\newblock In {\em Algebraic Methods in Graph Theory II}, pages 495--517.
  North-Holland, 1981.

\bibitem{lovasz2009matching}
L.~Lov{\'a}sz and M.~D. Plummer.
\newblock {\em Matching Theory}.
\newblock American Mathematical Society, 2009.

\bibitem{oxley2011matroid}
J.~Oxley.
\newblock {\em Matroid Theory}.
\newblock Oxford University Press, 2011.

\bibitem{robinson1980computational}
G.~Robinson and D.~Welsh.
\newblock The computational complexity of matroid properties.
\newblock In {\em Mathematical Proceedings of the Cambridge Philosophical
  Society}, volume~87, pages 29--45. Cambridge University Press, 1980.

\bibitem{schrijver2003combinatorial}
A.~Schrijver.
\newblock {\em Combinatorial Optimization: Polyhedra and Efficiency}.
\newblock Springer, 2003.

\end{thebibliography}

\end{document}